\documentclass[12pt]{amsart}

\usepackage{tgpagella}
\usepackage{euler}
\usepackage[T1]{fontenc}
\usepackage{amsmath, amssymb}
\usepackage[hidelinks]{hyperref}
\usepackage[english]{babel}
\usepackage{mathrsfs}
\usepackage{eucal}
\usepackage[all]{xy}
\usepackage{tikz}

\newtheorem{thm}{Theorem}
\newtheorem*{thm*}{Theorem}
\newtheorem{lem}[thm]{Lemma}

\newtheorem{prop}[thm]{Proposition}
\newtheorem*{prop*}{Proposition}

\newtheorem*{cor*}{Corollary}

\theoremstyle{definition}
\newtheorem{defn}[thm]{Definition}
\newtheorem*{defn*}{Definition}

\newtheorem*{question*}{Question}
\newtheorem*{Pquestion*}{Popa's question}

\newtheorem*{conv*}{Convention}

\def\cal{\mathcal}

\def\fin{\Subset}

\def\M{\mathcal{M}}

\makeatletter

\def\dotminussym#1#2{%
  \setbox0=\hbox{$\m@th#1-$}%
  \kern.5\wd0%
  \hbox to 0pt{\hss\hbox{$\m@th#1-$}\hss}%
  \raise.6\ht0\hbox to 0pt{\hss$\m@th#1.$\hss}%
  \kern.5\wd0}

\def \cirB{{}^{\circ}B}
\def \cirH{{}^{\circ}\mathbf{H}}

\def \fin{\operatorname{fin}}
\def \ns{\operatorname{ns}}
\def \st{\operatorname{st}}

\def \starR{{}^*\mathbb R}
\def \starC{{}^*\mathbb C}
\def \starN{{}^*\mathbb N}
\def \starX{{}^* X}

\def \bH{\mathbf{H}}

\def \dphi{|\phi\rangle}
\def \cdphi{{}^{\circ}|\phi\rangle}
\def \dpsi{|\psi\rangle}

\def \dpsii{|\psi_i\rangle}

\begin{document}


\title{Everettian mechanics with hyperfinitely many worlds}
\author{Jeffrey Barrett and Isaac Goldbring}
\thanks{Goldbring was partially supported by NSF grant DMS-2054477.}

\address{Department of Logic and Philosophy of Science\\University of California, Irvine, 765 Social Sciences Tower,
Irvine, CA 92697-5000}
\address{Department of Mathematics\\University of California, Irvine, 340 Rowland Hall (Bldg.\# 400),
Irvine, CA 92697-3875}
\email{j.barrett@uci.edu}
\urladdr{
https://faculty.sites.uci.edu/jeffreybarrett/}
\email{isaac@math.uci.edu}
\urladdr{http://www.math.uci.edu/~isaac}

\maketitle

\begin{abstract}
The present paper shows how one might model Everettian quantum mechanics using hyperfinitely many worlds. A hyperfinite model allows one to consider idealized measurements of observables with continuous-valued spectra where different outcomes are associated with possibly infinitesimal probabilities. One can also prove hyperfinite formulations of Everett's limiting relative-frequency and randomness properties, theorems he considered central to his formulation of quantum mechanics. This approach also provides a more general framework in which to consider no-collapse formulations of quantum mechanics more generally.
\end{abstract}

\section{Everettian quantum mechanics and many worlds}

Hugh Everett III (1956) (1957) presented pure wave mechanics as a solution to the quantum measurement problem encountered by the standard collapse theory.\footnote{Barrett (2019, 42--53 and 105--17) for a detailed description of the standard collapse formulation of quantum mechanics and the measurement problem and Barrett (2018) for an overview of Everett's project.} He characterized pure wave mechanics as the standard von Neumann-Dirac formulation of quantum mechanics but without any collapse of the quantum-mechanical state on measurement or any other time. All physical systems always evolve in a deterministic, linear way described by the standard quantum dynamics.

On Everett's interpretation of the quantum-mechanical state, the linear dynamics describes the universe as constantly splitting into branches or worlds corresponding to different measurement outcomes. How many worlds there are and the state of each world depends on how one understands the global state. Everett's description of the principle of the relativity of states allows for \emph{arbitrary} decompositions of the state, and he explicitly allowed for decompositions in terms of eigenstates of observables with continuous-valued spectra. His discussion of relative states in the context of von Neumann's account of a position measurement, for example, involves a superposition of a continuous number of branches, in each of which the object system has a different determinate position (1957, 180--2). For this reason, he was careful to allow for information measures and probability distributions over sets of unrestricted cardinality in his theory (1956, 86, 89--92). In brief, on Everett's own presentation there is a potentially uncountably infinite number of worlds depending on which decomposition of the state one considers.

While Everett almost always referred to branches rather than worlds (and the term ``world'' never appears in his written work), he explicitly endorsed there being an uncountable number of \emph{worlds} in a recorded discussion with colleagues in 1962. He had been invited to describe his formulation of quantum mechanics at a conference at Xavier University that had been convened to discuss the quantum measurement problem. After Everett had briefly described how the theory worked, the physicist Boris Podolsky said, ``It looks like we would have a non-denumerable infinity of worlds.'' To which Everett responded, ``Yes.''\footnote{For a transcript of Everett's description of his theory and the discussion that followed see Barrett and Byrne 2012, 270--79.}

In contrast, Bryce Dewitt, perhaps Everett's most energetic and effective proponent, held that since the correlations produced by measurement-like interactions are always only approximate for observables with continuous spectra, there must be at most a denumerable number of worlds.\footnote{See DeWitt (1971, 210--11) in DeWitt and Graham (1973) for his account of the measurement of observables with continuous spectra.} Indeed, DeWitt repeatedly suggested that the cardinally was large but finite, famously reporting that there were $10^{100+}$ constantly splitting worlds.\footnote{See DeWitt (1970, 33) and (1971, 179) for such claims.} Significantly, it is customary for many-worlds proponents simply to stipulate that there are at most a finite number of worlds for technical reasons and/or as a matter of conceptual convenience.\footnote{Sebens and Carroll (2015), for example, stipulate a finite number of worlds in order to assign self-location probabilities to worlds by means of a principle of indifference. Among other things, this avoids the technical problem one faces in assigning an unbiased prior over even just a countably infinite number of worlds using the standard mathematical model.}

In the context of deocoherence formulations of Everettian quantum mechanics, the question of the number of worlds is more subtle. As as particularly salient example, David Wallace has argued that since decoherence, the physical phenomena that tends to prevent interference between branches, comes in degree, there is no simple matter of fact about how many worlds there are at a time. Rather, how one individuates worlds on a decohering-worlds account depends on one's level of description and practical interests.\footnote{See Wallace (2012 99--102) for a brief introduction to the idea.} In order to prove the representation theorems that he wants for his decohering-worlds formulation of quantum mechanics, Wallace explicitly rules out the possibility of positive infinitesimal probabilities of the sort one would invariably encounter if one allowed for an infinite number of worlds.\footnote{See Wallace (2012, 227) and his discussion of decision theory in Appendix B.} But on the pragmatic view of emergent worlds that he has in mind, one arguably need never consider more than a finite number of worlds in any case.

Our aim here is not to provide a once-and-for-all account of worlds for Everettian quantum mechanics. Rather, it is to suggest a way of reconstructing Everett's original picture of how branches work and a way of thinking of non-denumerably many worlds and probability distributions over such worlds when it is convenient to do so. Specifically, a hyperfinite model allows one to consider idealized measurements of observables with continuous-valued spectra and it provides an intuitive set of tools for studying non-denumerable collections of worlds. It also allows one to provide elegant nonstandard reformulations of the two limiting properties that Everett took to be central to his formulation of quantum mechanics. And it provides an intuitive context that works very much like a finite-dimensional Hilbert space for representing no-collapse formulations of quantum mechanics more generally.

\section{worlds and probabilities}\label{worlds}

It will be helpful to start with a simple example of how pure wave mechanics describes an ideal measurement interaction where there are finitely many (here, two) possible measurement outcomes. Suppose that an observer $F$ measures the $x$-spin of a spin-$1/2$ system $S$ that begins in the state
\begin{equation}
\alpha |\!\uparrow_x\rangle_S + \beta |\!\downarrow_x\rangle_S.
\end{equation}
Let the state $|\mbox{``ready''}\rangle_F$ represent the state where $F$ is ready to observe the $x$-spin of $S$ and record the result in her notebook. The initial state of the composite system of $F$ and $S$ then is
\begin{equation}
|\mbox{``ready''}\rangle_F (\alpha |\!\uparrow_x\rangle_S + \beta |\!\downarrow_x\rangle_S).
\end{equation}
Since $|\mbox{``ready''}\rangle_F|\!\uparrow_x\rangle_S$ would evolve to 
$|\mbox{``$\uparrow_x$''}\rangle_F|\!\uparrow_x\rangle_S$ and $|\mbox{``ready''}\rangle_F|\!\downarrow_x\rangle_S$ would evolve to 
$|\mbox{``$\downarrow_x$''}\rangle_F|\!\downarrow_x\rangle_S$ on ideal measurements, the initial state above would evolve to the final post-measurement state
\begin{equation}
\alpha|\mbox{``$\uparrow_x$''}\rangle_F|\!\uparrow_x\rangle_S + \beta |\mbox{``$\downarrow_x$''}\rangle_F|\!\downarrow_x\rangle_S
\end{equation}
by the linearity of the standard quantum dynamics.

Note that there is no collapse of the quantum mechanical state here. Hence one avoids the problem of having to say how and when such a random event might occur. But while dropping the collapse dynamics from the standard collapse theory allows one to provide a manifestly consistent account of an idealized measurement interaction, it also immediately leads to two new problems. One involves how one explains determinate records, the other how one understands probabilities.

Since the final post-measurement state is not one where $F$ has any determinate measurement record at all on the standard eigenvalue-eigenstate rule for interpreting quantum-mechanical states, we need a new way to interpret states.\footnote{See Barrett (2019, 42--6) for a description of the standard interpretation of states.}  On the many-worlds interpretation of pure wave mechanics, easily the popular way of understanding Everett's theory, one understands each term of the final state above as corresponding to a physical world, one for each possible result. Inasmuch as the post-measurement state describes one observer with the result ``$\uparrow_x$'' and another observer with the result ``$\downarrow_x$'', this interpretation of the final state immediately solves the determinate-record problem.

While the standard collapse formulation of quantum mechanics encounters the quantum measurement problem and is hence ultimately unsatisfactory, it predicts precisely the right forward-looking probabilities for simple experiments like this one. Here it predicts that $F$ should expect to get the result ``$\uparrow_x$'' with probability $|\alpha|^2$ and the result ``$\downarrow_x$'' with probability $|\beta|^2$. That it predicts precisely the right forward-looking probabilities for such experiments is what makes the standard collapse theory one of the most successful physical theories ever. These are the empirical predictions that one would like to recapture in any satisfactory formulation of quantum mechanics. To address the probability problem, then, one needs to provide some way of understanding the standard quantum probabilities in pure wave mechanics, a deterministic theory where every physically possible measurement outcome is in fact fully realized as a world characterized by a branch of the quantum mechanical state.

Self-location probabilities, probabilities that represent one's epistemic uncertainty of finding oneself with a particular measurement result, provide the most promising way to understand probabilities in a many-worlds theory.\footnote{See Vaidman (2012) and (2014) for introductions to this approach.} But how one understands such probabilities depends on how one understands worlds and how one understands self-location uncertainty in those worlds.

Suppose one pictures $F$'s premeasurement world splitting into one copy for each possible measurement result when $F$ makes her measurement. On the face of it, she should expect (with probability 1) to find one future copy of herself in a world where she gets the result ``$\uparrow_x$'', and she should expect (also with probability 1) to find another future copy of herself in a world where she gets the result ``$\downarrow_x$''. But these are entirely the wrong quantum probabilities.

The way one understands probabilities in a many-worlds theory depends on the basic metaphysical picture one adopts for the worlds. Inasmuch as each of the future copies of $F$ have equal claim to being $F$ on a splitting-worlds view, she cannot make straightforward sense of standard forward-looking quantum probabilities as subjective degrees of belief regarding which world she will inhabit after she makes her measurement. That said, after she performs her measurement and subsequently inhabits a world where there is now a single determinate record, she can make perfectly good sense of quantum probabilities as synchronic self-location probabilities.

Suppose $F$ makes her measurement and suppose that her initial world has split into one with a copy of $F$ and her measuring device that records the result ``$\uparrow_x$'' and another with a copy of $F$ and her measuring device records the result ``$\downarrow_x$''. Consider one of the copies of $F$. Suppose that she has not yet looked at the result that the measuring device in her world has recorded. $F$ can now ask herself the perfectly coherent synchronic question of what her degree of belief should be that the recorded result in the world she inhabits right now is ``$\uparrow_x$'' and what her degree of belief should be that the recorded result in her world is ``$\downarrow_x$''.

While there is much to say about how one gets the values of such quantum probabilities in a many-worlds theory, we will suppose here that one has a formulation of the theory that stipulates that one should assign a subjective degree of belief equal to the norm-squared of the amplitude associated with one's post-measurement world conditional on one knows of one's premeasurement world.\footnote{How quantum statistics worked for Everett had to do with the limiting relative-frequency and limiting properties we discuss later in the paper. See Barrett (2019, 143--74) for a description of how his understanding of quantum statistics might be reconstructed.} In the present case, since we are supposing that $F$ begins ready to perform the described measurement, one wants one's full theory to stipulate $|\alpha|^2$ as the probability of getting ``$\uparrow_x$'' and $|\beta|^2$ the probability of getting ``$\downarrow_x$''. Here one is getting the standard quantum probabilities as synchronic self-location probabilities rather than forward-looking probabilities regarding her outcome. What $F$ \emph{can} say before performing her measurement is that each of her future selves will have post-measurement synchronic self-location probabilities given by the norm squared of the amplitude associated with each world.

One can get something that is arguably closer to standard forward-looking quantum probabilities by adopting a many-worlds theory where the worlds do not split. Suppose that there is one world for every possible history or thread through the branching structure generated by a series of subsequent measurements. On this picture, there are two copies of $F$ both before and after the measurement interaction on the simple example above. One copy will get the result ``$\uparrow_x$'' and the other will get the result ``$\downarrow_x$''. Here one can take quantum probabilities to be degrees of belief concerning which history or thread one inhabits. Even before performing her measurement, $F$ can coherently assign the standard forward-looking quantum probabilities to the outcome as they are just the synchronic self-location probabilities for her current inhabiting each of the two worlds.\footnote{See the discussion of many-threads theories in Barrett (2019, 184--87).}

While there is more to say about how one might individuate worlds and understand probability in even the simple case with two possible measurement results, this will serve for our present purposes. The hyperfinite approach that we develop here can be applied to either splitting or non-splitting worlds. To keep things simple, we will focus on splitting worlds.

The question we turn to now is how one might understand worlds and probabilities in the context of an idealized measurement of a physical observable with a continuous spectrum. The model involves a hyperfinite number of worlds, each associated with a possibly infinitesimal probability.

\section{A brief introduction to nonstandard methods}\label{NSA}

The task at hand requires a few basic tools. The following prerequisites concern the nonstandard methods we will use to characterize hyperfinitely many worlds.\footnote{See Goldblatt (1998) for a standard introductory text and Goldbring and Walsh (2019) for a description of potential applications. For a text on nonstandard methods written for physicists see Albeverio et. al. (1986).}

The basic tenet of nonstandard analysis is to extend every set $X$ under consideration to a \textbf{nonstandard extension} $\starX$ which satisfies the following two properties:  (1)  $\starX$ has the same (first-order) logical properties as $X$\footnote{This fact is often called the \textbf{transfer principle}.}, and (2)  $\starX$ contains new ``ideal'' elements not present in the original set $X$.\footnote{The ``number'' of these new ideal elements is controlled by the \emph{saturation level} of the nonstandard extension.}  For example, the field of \textbf{hyperreals} $\starR$ is a field extending the field $\mathbb R$ that shares the same logical properties as the reals $\mathbb R$ (such as, for example, being an ordered field) while containing new \textbf{infinitesimal} and \textbf{infinite elements}.  Every finite (that is, noninfinitesimal) element $r$ of $\starR$ is infinitely close to a unique standard real number, called the \textbf{standard part} of $r$, denoted $\st(r)$.  In the sequel, we let $\fin(\starR)$ denote the set of finite elements of $\starR$.  In general, for $r,s\in \starR$, we write $r\approx s$ if $r$ and $s$ are infinitely close to each other.

Since functions are certain kinds of sets, they also have nonstandard extensions. Specifically, a function $f:X\to Y$, identified with its graph, will be extended to a subset of $\starX\times {}^*Y$.  It is easy to check that this nonstandard extension is itself the graph of a function, which we abusively denote $f:\starX\to {}^*Y$ (as it is readily verified that it extends the original function $f$).

An important distinction between subsets of a nonstandard extension $\starX$ is the \textbf{internal} vs. \textbf{external} distinction.  Given any set $X$, we are entitled to consider the nonstandard extension ${}^*\mathcal{P}(X)$ of the powerset $\mathcal{P}(X)$ of $X$.  Under a natural identification, we may view elements of ${}^*\mathcal{P}(X)$ as actual subsets of ${}^*X$; the sets thus obtained are referred to as the internal subsets of $\starX$.  By the transfer principle, internal subsets of $\starX$ have the same logical properties as ordinary subsets of $X$.  For example, internal subsets of $\starR$ that are bounded above have a supremum.  This fact need not hold for external (that is, non-internal) subsets of $\starR$, such as the set of infinitesimal elements.  Subsets of $\starX$ of the form ${}^*Y$, where $Y$ is a subset of $X$, are internal subsets of $\starX$, but not every internal subset of $\starX$ is of this form.  The \textbf{internal definition principle} states that a set defined using internal parameters in a first-order way is once again internal.  Consequently, finite sets are always internal.

Using the internal definition principle, given $N\in \starN$, the set of elements of $\starN$ between $1$ and $N$ is an internal subset of $\starN$, suggestively denoted $\{1,\ldots,N\}$ (although it is not in fact finite if $N$ is infinite). If a (necessarily internal) set $E$ is in internal bijection with the internal set $\{1,\ldots,N\}$ for some $N\in \starN$, then $E$ is called \textbf{hyperfinite}.  Such an $N$ is automatically unique and is called the \textbf{internal cardinality} of $E$, suggestively denoted $|E|$.  Finite sets are hyperfinite and their internal cardinality agrees with the usual notion of cardinality. Hyperfinite sets share many of the intuitive properties of finite sets. That said, if $E$ is hyperfinite but infinite, then its \emph{actual} cardinality is at least the continuum and can be even larger depending on the saturation properties of the nonstandard model. Such will be the set of hyperfinitely many worlds corresponding to an observable with a continuous spectrum. By the transfer princple, an internal subset of a hyperfinite set is itself hyperfinite.  By considering the nonstandard extension of the summation function, one can make sense of hyperfinite sums $\sum_{x\in E}f(x)$, where $E$ is a hyperfinite set and $f$ is some internal function defined on $E$.

The Loeb measure construction will be used extensively throughout this paper.  Suppose that $E$ is an internal set and $\mu$ is an internal finitely additive probability measure defined on some internal algebra of internal subsets of $E$.  Note that $\mu$ takes values in ${}^*[0,1]$.  It can be shown that the associated (genuine) finitely additive measure $F\mapsto \st(\mu(F))$ (defined on the same algebra as $\mu$) can be extended to a genuine (that is, $\sigma$-additive) probability measure $\mu_L$ on the $\sigma$-algebra generated by the original internal algebra.  This measure is called the \textbf{Loeb measure} associated to $\mu$.  An example of such an internal finitely additive probability measure $\mu$ is the hyperfinite counting measure associated to a hyperfinite set:  if $E$ is hyperfinite and $F$ is an internal subset of $E$, then $\mu(F):=\frac{|F|}{|E|}$.

Internal Hilbert spaces play a central role in the present hyperfinite model.  An internal Hilbert space consists of an internal set $\bH$ equipped with an internal addition $+:\bH\times \bH\to \bH$, an internal scalar multiplication $\cdot:\starC\times \bH\to \bH$, and an internal inner product $\langle\cdot|\cdot\rangle:\bH\times\bH\to \starC$ satisfying the usual Hilbert space axioms.\footnote{Except that completeness gets replaced by internal completeness.}  For $\dphi,\dpsi\in \bH$, we write $\dphi\approx \dpsi$ if $\|\dphi-\dpsi\|\approx 0$.  By the transfer principle, every internal Hilbert space $\bH$ possesses an internal orthonormal basis.  If this internal orthonormal basis is actually hyperfinite, then we say that the internal Hilbert space $\bH$ is \textbf{hyperfinite-dimensional}.  In this case, if $(\dpsii)_{i\in I}$ is a (hyperfinite) orthonormal basis for $\bH$, then every element $\dphi\in \bH$ can be uniquely expressed as a hyperfinite sum $\dphi=\sum_{i\in I}\alpha_i\dpsii$, where each $\alpha_i\in \starC$.

\section{continuous-valued observables and the hyperfinite model}

\subsection{Hyperfinite-dimensional state spaces}

Everettian quantum mechanics for finite-dimensional state spaces as in the spin example above extends naturally to the hyperfinite-dimensional setting. 
Fix a hyperfinite-dimensional Hilbert space $\bH$ and fix an observable $B$ on $\bH$, which is an internally self-adjoint operator on $\bH$. The observable $B$ is what determines what worlds there are at a time. Depending on one's formulation of the theory, it might be a determinate-pointer observable or determinate-record observable, or it might be selected by decoherence interactions or other dynamical considerations. Once specified, there is an internal orthonormal basis $(\dpsii)_{i\in I}$ of eigenvectors of $B$ with corresponding eigenvalues $(\lambda_i)_{i\in I}$, where each $\lambda_i\in \starR$.  These eigenvalues correspond to the determine measurement outcomes when measuring the observable $B$.  The set $W$ of worlds thus coincides with the set $\{\dpsii \ : \ i\in I\}$ of eigenstates of $B$.  An arbitrary state is simply a unit vector $\dphi\in \bH$ and can be decomposed as a hyperfinite sum $\dphi=\sum_{i\in I}\alpha_i\dpsii$ with each $\alpha_i\in \starC$.  In this state, the probability of finding oneself in the world corresponding to $\dpsii$ is the hyperreal number $|\alpha_i|^2\in \starR$, which may be infinitesimal.\footnote{Such probabilities merely satisfy ``hypercountable'' additivity instead of the usual countable additivity.}

The linear dynamics for the hyperfinite setting is analogous to the finite setting.  Specifically, if one fixes an internally self-adjoint operator $T$ on $\bH$, the Hamiltonian for the system, then one can consider the unitary time evolution $V_t$ of the system, which is an internal one-parameter family of internal unitary operators on $\bH$ given by $V_t\dphi=e^{-itT}\dphi$.\footnote{There are a number of equivalent ways of interpreting the exponentiated operator here.  For example, one can view matrix exponentiation as a function $\bigcup_{n\in \mathbb N}M_n(\mathbb C)\to \bigcup_{n\in \mathbb N}M_n(\mathbb C)$ and then the above exponentiated operator is given by the nonstandard extension of this function.}  Note that these operators are defined for all $t\in \starR$. 

\subsection{Hyperfinite-dimensional models for standard state spaces}

We now consider how one might model the standard situation of an observable with continuous spectrum using a hyperfinite-dimensional state space as described in the previous subsection.\footnote{See Raab (2004) for a preliminary description of this example in the context of the standard collapse theory.}  Before doing so, we need some preliminaries.

\begin{defn}
Fix a Hilbert space $\cal H$ with dense subspace $\cal M$.  We say that a hyperfinite-dimensional subspace $\bH$ of ${}^*\cal H$ is \textbf{adapted to} $\cal M$ if $\cal M\subseteq \bH\subseteq {}^*\cal M$.
\end{defn}

The following is a routine application of saturation:

\begin{lem}
For any dense subspace $\cal M$ of a Hilbert space $\cal H$, there is a hyperfinite-dimensional subspace of ${}^*\cal H$ adapted to $\cal M$.\footnote{The uninitated reader may view this result as a more elaborate version of the statement that there exist positive infinitesimal real numbers.  The latter statement follows from an appropriate saturation assumption together with the fact that given any finite collection of positive real numbers, there is a positive real number smaller than all of them.  Here, given any finite subset of $\cal M$, there is a finite-dimensional subspace of $\cal \M$ containing all of them.}
\end{lem}

Suppose now that $\bH$ is adapted to $\cal M$.  Set $E_\bH:{}^*\cal H\to \bH$ to be the internal orthogonal projection map.  A routine ``overflow'' argument yields the following:

\begin{lem}
For all $\dphi\in \cal H$, $E_\bH\dphi\approx \dphi$.\footnote{The basic idea here is that the internal set of all $n\in \starN$ for which there exists $\dpsi\in \bH$ such that $\|\dphi-\dpsi\|<\frac{1}{n}$ contains the external set $\mathbb N$ (as $\bH$ contains the dense subspace $\cal M$ of $\cal H$).  Consequently, there must exist an infinite $N$ in this set, whence the lemma follows.}
\end{lem}

Suppose further that $T$ is an unbounded operator on $\cal H$ whose domain $D(T)$ contains $\cal M$.  Note that $T$ extends to a map $T:{}^*D(T)\to {}^*\cal H$.  Since $\bH\subseteq {}^*\cal M\subseteq {}^*D(T)$, we can consider the restriction of $T$ to $\bH$, namely $T\upharpoonright \bH:\bH\to {}^*\cal H$.  In order to obtain an operator on $\bH$, we must compose this latter map with the projection $E_\bH$.  Summarizing, we define the \textbf{natural extension} $T_\bH$ to be the internal operator on $\bH$ given by $T_\bH\dphi:=(E_\bH\circ T)\dphi$.  By the previous lemma, we have that $T_\bH\dphi\approx T\dphi$ for all $\dphi\in \cal M$.  The next lemma is easy to prove:

\begin{lem}
If $T$ as above is symmetric, then so is $T_\bH$ (whence it is internally self-adjoint).
\end{lem}

Crucial for us is the following (see Proposition 2 in Raab (2004)):

\begin{prop}\label{appev}
For each $\lambda\in \sigma(T)$, there is $\lambda'\in\sigma(T_\bH)$ (that is, an internal eigenvalue of $T_\bH$) such that $\lambda\approx \lambda'$.
\end{prop}

We now apply these preliminaries to the task at hand.  Fix a Hilbert space $\cal H$ (not necessarily separable), which is to represent the state space of our physical system.  Moreover, we fix an observable $A$, which is an unbounded self-adjoint operator on $\cal H$ with domain $D(A)$, and the Hamiltonian $H$ of our system, which is also an unbounded self-adjoint operator on $\cal H$ with domain $D(H)$.  Suppose further that $\cal M$ is a dense subspace of $\cal H$ that is contained in $D(A)\cap D(H)$.\footnote{This already presupposes that $D(A)\cap D(H)$ is a dense subspace of $\cal H$; we will add a further restriction in a moment.}  Fix a hyperfinite-dimensional subspace $\bH$ of ${}^*\cal H$ adapted to $\cal M$ and consider the natural extensions $A_\bH$ and $H_\bH$ of $A$ and $H$ respectively.   

One can connect the Everettian interpretation of the hyperfinite-dimensional system $(\bH,A_\bH,H_\bH)$ to an Everettian interpretation of the standard system $(\cal H,A,H)$ in a natural way. Proposition \ref{appev} already suggests that how this might work.  Indeed, for each $\lambda\in \sigma(A)$, a potential measurement result for the observable $A$, there is an eigenvalue (that is, a determinate measurement) for the internal observable $A_\bH$ that is infinitely close to $\lambda$, which we may interpret as an approximation to measuring $\lambda$ with infinite precision. But there are further conditions we would like to hold, namely:

\begin{enumerate}
     \item For standard states, the standard and internal unitary time evolutions should be infinitely close to one another.
    \item The probabilistic interpretations afforded by the hyperfinite model should agree (up to an infinitesimal difference) with the usual quantum mechanical probabilities associated to the standard system.  
   
\end{enumerate}

Before formulating a precise definition from these two conditions, we will say a bit more about each.

Regarding (1), for each $t\in \mathbb R$, set $U_t$ to be the unitary operator on $\cal H$ given by $U_t:=e^{-itH}$.\footnote{Of course now the exponentiated operator is defined using the spectral theorem.}  Also, for $t\in \starR$, we can consider the internal unitary operator $V_t$ on $\bH$ given by $V_t:=e^{-itH_{\bH}}$.\footnote{One might be tempted to simply consider the natural extension $(U_t)_\bH$ of $U_t$, but in general this need not be unitary.}  A precise formulation of condition (1) above would be that $U_t\dphi\approx V_t\dphi$ for all $\dphi\in D(H)$ and all $t\in \mathbb R$, that is, for a state in the domain of the Hamiltonian $H$ (for which $U_t$ does really solve the time-dependent Schr\"odinger equation with initial state $\dphi$), the standard and internal unitary evolutions of the state remain infinitely close for all (standard) time.

In connection with the previous paragraph, we say that an element $\dphi\in {}^*\cal H$ is \textbf{nearstandard}, and write $\dphi\in \ns({}^*\cal H)$, if there is $\dpsi\in \cal H$ such that $\dphi\approx \dpsi$.  In this case, we set $\st\dphi$ to denote this (necessarily unique) element of $\cal H$.

We now consider item (2).  For each $\dphi\in \cal H$ and each Borel subset $E\subseteq \mathbb R$, set $\mu^{A,\dphi}(E)$ to be the probability that a measurement of the observable $A$ in the state $\dphi$ yields a result in the set $E$.  Formally, $\mu^{A,\dphi}(E)=\langle \phi|P_E|\phi\rangle$, where $P_E$ is the projection-valued measure associated with the observable $A$ applied to the Borel set $E$.  On the other hand, for $\dphi\in \bH$ and internal $F\subseteq \starR$, set $\mu^{A_\bH,\dphi}(F)$ to be the internal probability that a measurement of the observable $A_\bH$ in the state $\dphi$ yields a result in the internal set $F$.  By transfer, this is calculated by the hyperfinite sum $\mu^{A_\bH,\dphi}(F)=\sum \{|\alpha_i|^2 \ : \ \lambda_i\in F\}$, where $\dphi=\sum_i \alpha_i\dpsii$ is the expansion of $\dphi$ in terms of the eigenbasis $\{\dpsii\ : \ i\in I\}$ fixed for $A_\bH$.  Note that $\mu^{A_\bH,\dphi}$ is indeed an internal finitely additive probability measure defined on the internal subsets of $\starR$.  

Given our discussion of item (1) above and the fact that a measurement of $A_\bH$ that belongs to $\st^{-1}(E)$ should be viewed as approximating (with infinite precision) a measurement of $A$ landing in $E$, a first guess as to what (2) might require is that, for each $\dphi\in \ns({}^*\cal H)$ and Borel subset $E\subseteq \mathbb R$, we have $\mu^{A,\st\dphi}(E)\approx\mu^{A_\bH,\dphi}(\st^{-1}(E))$.  The issue with this is that $\st^{-1}(E)$ is not generally an internal subset of $\starR$.  However, writing $E_n$ for the set of elements of $\starR$ within distance $\frac{1}{n}$ of ${}^*E$ and noting that each $E_n$ is internal, we see that $\st^{-1}(E)=\bigcap_{n\in \mathbb{N}^{>0}}E_n$ is Loeb measurable.  Thus, a correct formalization of (2) above would ask that $\mu^{A,\st\dphi}(E)=\mu^{A_\bH,\dphi}_L(\st^{-1}(E))$ for all $\dphi\in \ns({}^*\cal H)$ and all Borel subsets $E\subseteq \mathbb R$.

We summarize this discussion with a definition:

\begin{defn}
The hyperfinite-dimensional system $(\bH,A_\bH,H_\bH)$ is a \textbf{faithful model} of the standard system $(\cal H,A,H)$ if the following two conditions hold:
\begin{enumerate}
  \item For each $\dphi\in D(H)$ and each $t\in \mathbb R$, we have $U_t\dphi\approx V_t\dphi$.
    \item For each $\dphi\in \ns({}^*\cal H)$ and each Borel subset $E\subseteq \mathbb R$, we have $$\mu^{A,\st\dphi}(E)=\mu^{A_\bH,\dphi}_L(\st^{-1}(E)).$$
\end{enumerate}
\end{defn}

Note that in a faithful model, by setting $E:=\sigma(A)$, we have that, for any state $\dphi\in \ns({}^*\cal H)$, with $\mu^{A_\bH,\dphi}_L$-probability $1$, a measurement result of $A_\bH$ in the state $\dphi$ will yield a definite measurement that is infinitely close to an element of $\sigma(A)$.

In order to obtain a faithful model, a further technical assumption on $\cal M$ must be made.  This technical assumption does indeed hold in many cases of interest, as we will discuss below.

We set $$\fin(\bH):=\{\dphi\in \bH \ : \ \|\dphi\|\in \fin(\starR)\}.$$  For $\dphi\in \fin(\bH)$, we note that $\dphi\mapsto \st\|\dphi\|$ is a semi-norm on $\fin(\bH)$.  We also set $\operatorname{mon}(\bH):=\{\dphi\in \fin(\bH)\ : \ \|\dphi\|\approx 0\}$ and $\cirH:=\fin(\bH)/\operatorname{mon}(\bH)$, which is then a normed space under the norm induced by the above seminorm on $\fin(\bH)$.  If $ \dphi\in \fin(\bH)$, we denote its class in $\cirH$ by $\cdphi$.  For $\dphi,\dpsi\in \fin(\bH)$, note that $\langle \phi|\psi\rangle\in \fin(\starC)$ and thus the internal inner product on $\bH$ descends to an inner product on $\cirH$ defined by $\langle {}^{\circ}\phi|{}^{\circ}\psi\rangle:=\st\langle \phi|\psi\rangle$.  It follows from saturation that $\cirH$ is actually a Hilbert space, called the \textbf{nonstandard hull} of $\bH$.  Moreover, since $\cal M$ is dense in $\cal H$ and $E_\bH \dphi\approx \dphi$ for all $\dphi\in \cal H$, it follows that the map $\dphi\mapsto {}^{\circ}\dphi:\cal H\to \cirH$ is an embedding of Hilbert spaces; in what follows, one identifies $\cal H$ as a subspace of $\cirH$ via this embedding.

Suppose now that $B$ is an internally self-adjoint operator on $\bH$. In Raab (2004), the \textbf{nonstandard hull} of $B$ is defined to be a certain self-adjoint operator $\cirB$ on $\cirH$ whose precise definition is given in terms of spectral theory.  In the easy case that $B$ is finitely bounded, that is, when the internal operator norm $\|B\|$ of $B$ is a finite element of $\starR$ (which rarely happens in quantum mechanics), we have that $B(\fin(\bH))\subseteq \fin(\bH)$, whence one can define $\cirB\cdphi:={}^{\circ}(B\dphi)$. (This definition does not even require $B$ to be self-adjoint.)

Given a self-adjoint operator $A$ on $\cal H$ and an internally self-adjoint operator $B$ on $\bH$, we say that $B$ is a \textbf{hull extension} of $A$ if $D(A)\subseteq D(\cirB)$ and $A\dphi=\cirB \dphi$ for all $\dphi\in D(A)$.\footnote{Raab calls such a $B$ a nonstandard extension of $A$, but we find this nomenclature confusing.}  The reason for being interested in hull extensions is given by the following result, which follows from results in Raab (2004) and with further details presented in Goldbring (2021):

\begin{thm}
Suppose that $A_\bH$ and $H_\bH$ are hull extensions of $A$ and $H$ respectively.  Then the hyperfinite-dimensional system $(\bH,A_\bH,H_\bH)$ is a faithful model of the standard system $(\cal H,A,H)$.
\end{thm}


Raab (2004) also provides a condition for when the natural extension is a hull extension.  Recall that a \textbf{core} of an unbounded self-adjoint operator $A$ is a dense subspace $\cal M$ of its domain on which $A\upharpoonright \cal M$ is essentially self-adjoint.

\begin{thm}
Suppose that $A$ is a self-adjoint operator on $\cal H$ and that $\cal M$ is a core of $A$.  Then for any hyperfinite-dimensional subspace $\bH$ of ${}^*\cal H$ that is adapted to $\cal M$, the natural extension $A_\bH$ is a hull extension of $A$.
\end{thm}

Summarizing the relationship between the standard and nonstandard models, we have:

\begin{thm}
Suppose that $(\cal H,A,H)$ is a standard system for which the observable $A$ and the Hamiltonian $H$ share a common core $\cal M$.  Then for any hyperfinite-dimensional subspace $\bH$ of ${}^*\cal H$ adapted to $\cal M$, the hyperfinite-dimensional system $(\bH,A_\bH,H_\bH)$ is a faithful model of $(\cal H,A,H)$.
\end{thm}

There are many instances in which the common core assumption is satisfied.  For example, whenever the Hamiltonian $H$  is a ``small perturbation'' of the observable $A$, then the common core assumption holds.  Relevant theorems along these lines are \emph{W\"ust's theorem} and the \emph{Kato-Rellich theorem}.\footnote{See Chapter X of Reed and Simon (1975) for more details.}  We will study particular instances of these results when we study the motion of a particle in $\mathbb R^n$ in Section \ref{examplemotion}.

The hyperfinite model arguably comes with surplus expressive power. In addition to the sort of states that one would want to be able to represent, the representation allows for nonstandard states with curious properties. There are, for example, states $\dphi\in \bH$ such that $\mu^{A_\bH,\dphi}_L(\st^{-1}(\sigma(A))<1$. These are states for which there is a positive probability of getting a measurement result $\lambda\in \sigma(A_\bH)$ that may not belong to $\st^{-1}(\sigma(A))$.  Such a $\lambda$ may be a finite hyperreal, in which case one would have a nonzero probability of measuring a value that is infinitely close to a real value which itself could never have been a result of measuring $A$.\footnote{If $\sigma(A)=\mathbb R$, as in the case of position or momentum, this possibility cannot arise.}  On the other hand, such a $\lambda$ may be infinite. In this case the corresponding eigenstate $\dpsi$ would correspond to an infinite expectation value for the observable. There is nothing inherently wrong with allowing for such a possibility when it makes good physical sense. But sometimes it doesn't.\footnote{Such phenomena have been discussed in other nonstandard treatments of quantum mechanics. See, for example, Benci et. al. (2019).} 

That said, there is little reason to worry about excessive expressive power of the hyperfinite model in the context of a no-collapse formulation like Everettian quantum mechanics. Here the nearstandard vector representing a nearstandard global state always evolves infinitely close to the usual deterministic unitary evolution and hence stays nearstandard. Consequently, states that make good physical sense continue to make good physical sense under the dynamics. Further, what worlds there are and what their states are at each time is fully determined by the linear decomposition of nearstandard state in terms of the eigenstates of one's specified observable.\footnote{Compare this with Raab's (2004) formulation of the standard collapse theory using nonstandard methods where the dynamics may require collapses to nonstandard eigenstates where there is no clear physical interpretation.}






\section{Nonstandard formulations of Everett's limiting properties}

Everett considered the statistical limiting properties of pure wave mechanics to be centrally important to his interpretation of the theory. He discussed this point at a conference on the foundations of quantum mechanics at Xavier University in 1962. While the participants were most interested in talking about branching worlds, Everett wanted to talk about his relative frequency and randomness results. He began an extended monologue on the topic:
\begin{quote}
I'd like to make one final remark here. Imagine a very large series of experiments made by an observer.  With each observation, the state of the observer splits into a number of states, one for each possible outcome, and correlated to the outcome. Thus the state of the observer is a constantly branching tree, each element of which describes a particular history of observations. Now, I would like to assert that, for a ``typical'' branch, the frequency of results will be precisely what is predicted by ordinary quantum mechanics. (Barrett and Byrne 2012, 274--5)
\end{quote}
To make this claim, Everett explained, one needs a measure over branches. He chose the norm-squared coefficient measure on the grounds that is had a number of salient formal properties. He took his most significant achievement to be in showing that, in this measure, measure one of the branches in the determinate-record basis will exhibit the standard quantum statistics in the limit as an infinite number of measurements are made. 

Everett had briefly sketched his argument for his \emph{relative frequency} property in the long version of his Ph.D.\ thesis (1956, 126--7) and in the much shorter version of his thesis that he defended and was subsequently published as a standalone research paper (1957, 190--4). He had also, yet more briefly, sketched an argument for a \emph{randomness} property (1956, 127--8). These two properties together say that measure one of the branches will exhibit random sequences of results exhibiting the usual quantum statistics in the limit as an infinite number of measurements are made.\footnote{See Hartle (1968) for an early discussion of the relative frequency property, and and Farhi, Goldstone, and Gutmann (1989) for a proof. See Barrett (1999) for a discussion of the experimental setup and both the relative-frequency and randomness properties.}  Hence, Everett concluded, ``all predictions of the usual theory will appear to be valid to the observer in almost all observer states'' (1957, 194). This is what it meant to him for his formulation of pure wave mechanics to be empirically adequate.\footnote{See Barrett (2015) for a discussion of how Everett thought of empirical adequacy.}

To see how this works in the context of the hyperfinite model, fix an infinite element $K$ of $\starN$ and let $\bH:=\bigotimes_{i=1}^K \starC^2$ denote the hyperfinite-dimensional space that is the internal tensor product of $K$ copies of $\starC^2$. This space allows one to represent the composite system of the measuring device and the first $K$ object systems after $K$ sequential measurements of $x$-spin.  Suppose that each particle is prepared in state (1) above.  We consider the element $|\psi\rangle:=\sum_{x\in 2^K}\alpha^{i_x}\beta^{K-i_x}|x\rangle$ of $\bH$; here, $i_x$ denotes the internal cardinality of the set of those $j=1,\ldots,K$ such that $x(j)=1$ (which represents $x$-spin up).  Set $\lambda_x:=\frac{i_x}{K}$ for the relative frequency of achieving $x$-spin up according to the branch $x$.  We let $B:\bH\to\bH$ denote the internal self-adjoint operator defined by $B|x\rangle:=\lambda_x|x\rangle$.  Note that $B$ is internally bounded.

Let $\mu^{|\psi\rangle}$ denote the internal scalar-valued spectral measure corresponding to $|\psi\rangle$, that is, for any internally Borel subset $E\subseteq \starR$, $\mu^{|\psi\rangle}(E):=\sum_{\lambda_x\in E}|\alpha^{i_x}\beta^{K-i_x}|^2$.  We let $\mu_L^{|\psi\rangle}$ denote the corresponding Loeb measure.  We also let $\cirB:\cirH\to \cirH$ denote the corresponding nonstandard hull map (as defined in the previous section) with scalar-valued spectral measure $\nu^{{}^{\circ}|\psi\rangle}$.

The following theorem is the nonstandard reformulation of Everett's relative frequency theorem:\footnote{This theorem follows the same line of argument as the theorem for the standard case. See Barrett (1999, 100--4).}

\begin{thm}
Using the above notation, we have $\mu_L^{|\psi\rangle}(\{x\in 2^K \ : \ \lambda_x\approx \alpha^2\})=1$, whence $\nu^{{}^{\circ}|\psi\rangle}(\{\alpha^2\})=1$.  Consequently, ${}^{\circ}|\psi\rangle$ is an eigenvector of $\cirB$ with eigenvalue $\alpha^2$.
\end{thm}

The argument for the nonstandard interpretation of Everett's randomness theorem is relatively straightforward. Fix infinite $K\in \starN$ and consider the ``restriction map'' $f:2^K\to 2^{\mathbb N}$ which restricts a hyperfinite sequence of $0$'s and $1$'s to its initial segment indexed by standard natural numbers.  Consider $2^K$ equipped with the Loeb measure $\mu_L$ associated with the internal product measure on $2^K$ while we consider $2^{\mathbb N}$ with its usual Lebesgue measure (that is, infinite product measure) $m$.  It is not hard to verify that the pushforward of $\mu_L$ along $f$ is precisely $m$.  Consequently, if we repeat our measurement $K$ times as above, then the total measure, as determined by $\mu_L$, that our sequence $x$ belongs to $f^{-1}(E)$ for any Lebesgue null subset $E$ of $2^{\mathbb N}$ is $0$.  In particular, if $E$ is the set of nonrandom sequences, with respect to any standard criterion for what it means for such a sequence to be nonrandom, then $E$ will be countable and thus be Lebesgue null. Hence almost all sequences in measure $\mu_L$ will satisfy any standard criterion for being random.\footnote{See Barrett (1999, 104--5) for this theorem in the standard context and Barrett and Huttegger (2019) for an extended discussion of quantum randomness.} 

As a simple concrete example of the relative frequency and randomness properties, consider an observer who repeats an $x$-spin measurement on an $\omega$ sequence of particles each prepared in state (1) above. The two basic limiting properties in the context of the hyperfinite model entail that measure one of the branches, in measure $\mu_L$, will exhibit the standard quantum relative frequencies with $|\alpha|^2$ proportion of the results being ``$\uparrow_x$'' and $|\beta|^2$ proportion of the results being ``$\downarrow_x$'' and with no computationally specifiable pattern. Since measure one of the worlds will exhibit the standard quantum statistics with randomly distributed results, the standard quantum predictions will hold in a typical world (in the sense of typical given by measure $\mu_L$). 

We now prove a continuous spectrum version of Everett's relative frequency theorem.

Suppose that the hyperfinite-dimensional system $(\bH,A_\bH,H_\bH)$ is a faithful model of $(\cal H,A,H)$.  We adopt the same notation as in the previous section.  Fix $\dphi\in \ns({}^*\mathcal{H})$ and infinite $K\in \starN$.  Let $\zeta^{K,\dphi}$ denote the internal measure on $I^K$ obtained by taking the internal product measure of the measure $\mu^{A_\bH,\dphi}$.  (This is the internal notion of typicality here.)  Also, let $\eta_L^{K}$ denote the Loeb measure on $[1,K]$ coming from counting measure as described in Section \ref{NSA}.  For each $y\in I^K$ and Borel set $F\subseteq \mathbb R$, let $X_{F,y}:=\{j\in [1,K] \ : \lambda_{y(j)}\in \st^{-1}(F)\}$.  Note that each $X_{F,y}$ is a Loeb measurable subset of $[1,K]$.

\begin{thm}\label{extendedeverett}
Fix $\dphi\in \ns({}^*\cal H)$.  Then for sufficiently large infinite $K\in \starN$, we have the following:  for every Borel set $F\subseteq \mathbb R$, we have that $\eta_L^{K}(X_{F,K,y})=\mu^{A,\st\dphi}(F)$ for $\zeta_L^{K,\dphi}$-almost all $y\in I^K$.  
\end{thm}

\begin{proof}
Fix infinitesimal $\epsilon>0$.  For $K\in \starN$, $y\in I^K$, and internally Borel $E\subseteq \starR$, set $$X_{E,K,y}:=\{j\in [1,K] \ : \ \lambda_{y(j)}\in E\}.$$  Note that each $X_{E,K,y}$ is internal.  Furthermore, let $Z_{K}$ consist of those $y\in I^K$ such that $|\frac{|X_{E,K,y}|}{K}-\mu^{A_\bH,\dphi}(E)|<\epsilon$ for all internally Borel sets $E\subseteq \starR$.  Then by transferring the classical Everett limiting theorem, there is $K_0\in \starN$ such that, for all $K\geq K_0$, we have that $\zeta^{K,\dphi}(Z_K)>1-\epsilon$.  For each Borel set $F\subseteq \mathbb R$ and each $m\in\starN$ (finite or infinite), let $F_m:=\{\lambda\in \starR \ : \ \operatorname{dist}(\lambda,{}^*F)<\frac{1}{m}\}$.  Since $X_{F,K,y}=\bigcap_{m=1}^\infty X_{F_m,K,y}$, we have that $\eta^K_L(X_{F,K,y})\approx \eta^K(X_{F_M,K,y})$ for any infinite $M\in \starN$.  If $y\in Z_K$, then $\eta^K(X_{F_M,K,y})\approx \frac{|X_{F_M,K,Y}|}{K}\approx \mu^{A_\bH,\dphi}(F_M)$.  Since $\st^{-1}(F)=\bigcap_{m=1}^\infty F_m$, we have that $\mu^{A,\st\dphi}(F)=\mu^{A_\bH,\dphi}_L(\st^{-1}(F))\approx \mu^{A_\bH,\dphi}(F_M)$.

Putting this all together, we have that $\eta^K_L(X_{F,K,y})=\mu^{A,\st \dphi}(F)$ for all $y\in Z_K$.  Since $\zeta^{K,\dphi}_L(Z_K)=1$, we may conclude.
\end{proof}

In the statement of the theorem, how large $K$ needs to be chosen may depend on the state $\dphi\in \ns({}^*\cal H)$, that is, using the notation from the proof, the value of $K_0$ may depend on $\dphi$.  Setting $K_{\dphi}$ for this value of $K_0$ dependent on $\dphi$, it is readily verified that $K_{\dphi}=K_{\st\dphi}$.  Consequently, if the nonstandard extension is sufficiently saturated, then by choosing $K_0\in \starN$ greater than $K_{\dphi}$ for all $\dphi\in \cal H$, we see that the conclusion of the theorem holds for any $K\geq K_0$, \emph{independent of the choice of} $\dphi\in \ns({}^*\cal H)$.

In the next section, we will see an illustration of the limiting property just proven in the case of an observable with continuous spectrum. 



\section{An example:  motion of a particle in $\mathbb R^n$}\label{examplemotion}

In this section, we show how the hyperfinite representation works for a particular concrete example, namely for the motion of a single particle in $\mathbb R^n$.  In this case, $\cal H:=L^2(\mathbb R^n)$ and our observable $A$ under consideration will be $X_j$, the operator which multiplies a function by the variable $x_j$, for some $j=1,\ldots,n$.  This operator corresponds to observing the $j^{\text{th}}$-coordinate of the position of the particle.  The Hamiltonian for our system is given by $H:=-\frac{\hbar}{2m}\Delta+V(\vec X)$, where $\Delta$ is the $n$-dimensional Laplacian, $\vec X=(X_1,\ldots,X_n)$ is the vector consisting of the various position operators, and $V:\mathbb R^n\to \mathbb R$ is some function (whence $V(\vec X)$ represents the potential energy).   

We briefly discuss some natural conditions on which $A$ and $H$ satisfy the common core assumption.  A theorem of Kato (which follows from the aforementioned Kato-Rellich Theorem) states that if $n\leq 3$ and $V:\mathbb R^n\to \mathbb R$ belongs to $L^2(\mathbb R^n)+L^\infty(\mathbb R^n)$, then any core for $\Delta$ is a core for the Hamiltonian $H$.  A more general version of Kato's theorem holds for dimension $n\geq 4$ if the assumption $V\in L^2(\mathbb R^n)+L^\infty(\mathbb R^n)$ is replaced with $V\in L^p(\mathbb R^n)+L^\infty(\mathbb R^n)$ for some $p>\frac{n}{2}$.  Since $C^\infty_c(\mathbb R^n)$, the space of smooth functions on $\mathbb R^n$ with compact support, is a core for $\Delta$ as well as any $X_j$, we thus have many natural potential functions for which our common core assumption holds.\footnote{See Chapter X of Reed and Simon (1975) for more examples of when the common core assumption holds.}

From now on, we fix a potential function $V$ for which $A$ and $H$ share a common core $\cal M$.  We fix a hyperfinite-dimensional subspace $\bH$ of ${}^* \cal H$ that is adapted to $\cal M$.  Suppose that $(\dpsii)_{i\in I}$ is an internal orthonormal eigenbasis for $A_\bH$ with corresponding eigenvalues $(\lambda_i)_{i\in I}$.  Consequently, these $\lambda_i\in \starR$ represent the possible measurement outcomes when measuring $A_\bH$.  According to Proposition \ref{appev} above, for each $\lambda\in \sigma(A)=\mathbb R$, there is some $i\in I$ such that $\lambda\approx \lambda_i$.  That is, the measurement outcome $\lambda_i$ for $A_\bH$ should be thought of as an approximate measurement (with infinite precision) corresponding to a measurement of the $j^{\text{th}}$ coordinate of the particle being $\lambda$.  Each $\dpsii$, being an actual eigenvector of $A_\bH$, is such that $A_\bH\dpsii-\lambda_i\dpsii$ is orthogonal to every vector in $\bH$.  In particular, since $\bH$ contains $\cal M$, which is a dense subspace of $L^2(\mathbb R)$, we see that $\dpsii$ must behave like a Dirac function with $j^{\text{th}}$ coordinate function positioned at $\lambda_i$.

Fix some initial (standard) state $\dphi\in \cal H=L^2(\mathbb R)$ of the particle and write $\dphi=\sum_{i\in I}\alpha_i \dpsii$, where each $\alpha_i\in \starC$ (whence $\sum_{i\in I}|\alpha_i^2|=1$).  Suppose we then let the particle evolve according to the internal unitary evolution $V_t:=e^{-itH_\bH}$.  Given some $t\in \starR$, we may then consider the internal state of the system at time $t$, namely 
$$|\phi(t)\rangle:=V_t\dphi=V_t\left(\sum_{i\in I}\alpha_i\dpsii\right)=\sum_{i\in I}\beta_i \dpsii.$$ In the last step, we have rewritten the evolved state in terms of the eigenbasis for $A_\bH$.  Recall also that for standard time $t\in \mathbb R$, the above $\phi(t)\rangle$ will be infinitely close to the standard evolution $U_t|\phi\rangle$ of our original state $\dphi$.  

Suppose an observer now makes a measurement. The relative observer whose system state is $\dpsii$ would get the determinate result $\lambda_i$.  For a $\mu^{A_\bH,|\psi(t)\rangle}_L$-measure $1$ set of worlds, such $\lambda_i$ belongs to $\fin(\mathbb R)$ and represents an approximate measurement of the $j^{\text{th}}$ position of the particle being $\st(\lambda_i)$, where this approximate measurement has infinite precision.  Moreover, by Theorem \ref{extendedeverett} above, if one repeats this measurement a sufficiently large hyperfinite number of times, the relative frequency of landing infinitely close to any particular Borel set agrees with the standard quantum statistical measure $\mu^{A,|\st\phi(t)\rangle}$ of that Borel set. Further, the particular sequence of results will pass any standard test for randomness for almost all relative sequences in the norm-squared measure.  Finally, the global state $|\phi(t)\rangle$ evolves according to the internal dynamics $V$ at all times with the local states of the worlds determined by the nonstandard decomposition of the state.


\section{discussion}

The present hyperfinite model has a number of features that Everett would have found compelling. It allows one to understand the global state as a set of branches corresponding to the eigenvalues of a continuous-valued observable. The measure associated with the set of branches behaves in a natural way akin to what one might expect from a probability measure over a finite set. Indeed, the hyperfinite-dimensional state space behaves very much like a finite-dimensional state space. This makes operations like forming linear combinations of states and calculating probabilities particularly intuitive. The model gives values infinitesimally close to the standard quantum probabilities when the probabilities are finite. It is also possible for a branch to be associated with an infinitesimal probability. Conveniently, these sum just like finite probabilities. And the hyperfinite model allows for a straightforward formulation of the standard unitary quantum dynamics.

There is a close connection between the standard and the hyperfinite models. Specifically, as we have proven here, the latter provides a faithful representation of the former. The standard model, then, guides one in understanding the physical content of the nonstandard model. While the hyperfinite model arguably has more expressive power than one needs, nearstandard states remain nearstandard under the dynamics and hence continue to make good physical sense.

In addition to providing an intuitive picture of hyperfinitely many worlds, we have shown how one might prove versions of Everett's statistical limiting properties in the hyperfinite model. The result is that the norm-squared measure of the resulting branches that will exhibit the standard random quantum statistics in the limit as one performs an unbounded sequence of measurements is one. Hence, the statistical predictions of the standard collapse formulation of quantum mechanics will appear to be valid in almost all of the hyperfinite many worlds in this measure.

In brief, the hyperfinite model provides a way to think of a set of worlds that represent a continuous spectrum of measurement outcomes and for which one can prove versions of his limiting results. In this regard, it satisfies the conditions that Everett stipulated at the Xavier conference. Finally, it provides an intuitive framework in which to consider no-collapse formulations of quantum mechanics more generally.


\newpage

\begin{center} 
\large{Bibliography}
\end{center}

\noindent
Albeverio, S., J.E. Fenstad, R. Hoegh-Krohn, and T. Lindstrom (1986), ``Nonstandard methods in Stochastic Analysis and Mathematical Physics'', Dover Publications.

\vspace{.5cm}
\noindent
Barrett, Jeffrey A.\ (2020) \emph{The Conceptual Foundations of Quantum Mechanics}, Oxford University Press.

\vspace{.5cm}
\noindent
Barrett, Jeffrey A. (2015) ``Pure Wave Mechanics and the Very Idea of Empirical Adequacy,'' \emph{Synthese} 192(10): 3071-–3104.

\vspace{.5cm}
\noindent
Barrett, Jeffrey A. (2011a) ``On the Faithful Interpretation of Pure Wave Mechanics'' \emph{British Journal for the Philosophy of Science} 62 (4): 693--709. 

\vspace{.5cm}
\noindent
Barrett, Jeffrey A. (2011b) ``Everett's Pure Wave Mechanics and the Notion of Worlds'' \emph{European Journal for Philosophy of Science} 1 (2):277--302.

\vspace{.5cm}
\noindent
Barrett, Jeffrey A. (1999) \emph{The Quantum Mechanics of Minds and Worlds}, Oxford: Oxford University Press.

\vspace{.5cm}
\noindent
Barrett, Jeffrey A. and Peter Byrne (eds) (2012) \emph{The Everett Interpretation of Quantum Mechanics: Collected Works 1955-1980 with Commentary}, Princeton: Princeton University Press.

\vspace{.5cm}
\noindent
Barrett, J.\ A.\ and S.\ Huttegger (2019) ``Quantum Randomness and Underdetermination,'' forthcoming in \emph{Philosophy of Science}. https://doi.org/10.1086/708712.

\vspace{.5cm}
\noindent
Benic, V., L.\ Luperi Baglin, and K.\ Simonov, ``Infinitesimal and infinite numbers as an approach to quantum mechanics'' \emph{Quantum} (2019).

\vspace{.5cm}
\noindent
Bohm, David (1952) ``A Suggested Interpretation of Quantum Theory in Terms of `Hidden Variables','' Parts I and II, {\em Physical Review\/} 85, 166--179, 180--193.

\vspace{.5cm}
\noindent
DeWitt, Bryce S.\ (1971) ``The Many-Universes Interpretation of Quantum Mechanics'' in B.\ D.\ 'Espagnat (ed.) \emph{Foundations of Quantum Mechanics}. New York: Academic Press. Reprinted in DeWitt and Graham (1973) pp.\ 167--218.

\vspace{.5cm}
\noindent
DeWitt, Bryce S.\ and Neill Graham (eds.) (1973) \emph{The Many-Worlds Interpretation of Quantum Mechanics}. Princeton: Princeton University Press.

\vspace{.5cm}
\noindent
Everett, Hugh III (1956) ``The Theory of the Universal Wave Function.'' In Barrett and Byrne (eds) (2012, 72--172).

\vspace{.5cm}
\noindent
Everett, Hugh III (1957) `` `Relative State' Formulation of Quantum Mechanics,'' Reviews of Modern Physics, 29: 454--462. This is the published version of Everett's short Ph.D. thesis. In Barrett and Byrne (eds) (2012, 173--96).

\vspace{5mm}
\noindent
Farhi, E., J.\ Goldstone, S.\ Gutmann (1989) ``How Probability Arises in Quantum Mechanics,'' {\em Annals of Physics} 192(2) 368--382.

\vspace{5mm}
\noindent
Goldblatt, R. (1998) ``Lectures on the hyperreals'', \emph{Graduate Texts in Mathematics}

\vspace{5mm}
\noindent
Goldbring, I. (2021) ``A nonstandard proof of the spectral theorem for unbounded self-adjoint operators'', To appear in \emph{Expositions Mathematicae}.

\vspace{5mm}
\noindent
Goldbring, I. and S. Walsh (2019), ``An invitation to nonstandard analysis and its recent applications'' \emph{Notices of the AMS}, 66: 842-851.

\vspace{5mm}
\noindent
Hartle, J.\ B.: (1968) ``Quantum Mechanics of Individual Systems,'' {\em American Journal of Physics\/}, 36(8) 704--12.

\vspace{5mm}
\noindent
Raab, A. (2004) ``An approach to nonstandard quantum mechanics'', \emph{Journal of Mathematical Physics} 45: 4791--4809.

\vspace{5mm}
\noindent
Reed, M. and B. Simon (1975) ``Methods of modern mathematical physics, Volume 2'', Elsevier.

\vspace{.5cm}
\noindent
Saunders, Simon; Jonathan Barrett; Adrian Kent; David Wallace (eds) (2010) \emph{Many Worlds?: Everett, Quantum Theory, and Reality}, Oxford: Oxford University Press.

\vspace{.5cm}
\noindent
Sebens, C.\ and S.\ Carroll (2015) ``Self-Locating Uncertainty and the Origin of Probability in Everettian Quantum Mechanics'' arXiv:1405.7577v2 [quant-ph] 10 Jan 2015.

\vspace{.5cm}
\noindent
von Neumann, J.\ (1955) \emph{Mathematical Foundations of Quantum Mechanics}. Princeton: Princeton University Press. Translated by R. Beyer from \emph{Mathematische
Grundlagen der Quantenmechanik}. Berlin: Springer (1932).

\vspace{.5cm}
\noindent
Vaidman, Lev (2014) ``Many-Worlds Interpretation of Quantum Mechanics,'' \emph{Stanford Encyclopedia of Philosophy}. Fri Jan 17, 2014 main revision.

\vspace{.5cm}
\noindent
Vaidman, Lev (2012) ``Probability in the Many-Worlds Interpretation of Quantum Mechanics,'' in Yemima Ben-Menahem and Meir Hemmo (eds.), \emph{Probability in Physics (The Frontiers Collection XII)} Berlin Heidelberg: Springer, pp. 299--311.

\vspace{.5cm}
\noindent
Wallace, David (2012) \emph{The Emergent Multiverse: Quantum Theory according to the Everett Interpretation}, Oxford: Oxford University Press.

\end{document}